\documentclass[11pt,a4paper]{article}

\usepackage[latin1]{inputenc}
\usepackage{amssymb,amsmath,amsthm}
\usepackage{mathrsfs}
\usepackage{epsfig}

\theoremstyle{plain}
\newtheorem{thm}{Theorem}[section]
\newtheorem{lem}[thm]{Lemma}

\newtheorem{prop}[thm]{Proposition}

\newcommand{\real}{\ensuremath {\mathbb R} }

\newcommand{\nat}{\ensuremath {\mathbb N} }

\newcommand{\mbf}[1] {\text{\boldmath$#1$}}
\newcommand{\remove}[1] {}

\newcommand{\st}{\;:\;}

\newcommand{\ex} {{\bf E}}
\newcommand{\pr} {{\bf Pr}}

\newcommand{\tp} {{\mathtt{p}}}
\newcommand{\tn} {{\mathtt{n}}}
\newcommand{\tps} {{\mathtt{ps}}}
\newcommand{\tns} {{\mathtt{ns}}}
\newcommand{\tpu} {{\mathtt{pu}}}
\newcommand{\tnu} {{\mathtt{nu}}}
\newcommand{\tnsf} {{\mathtt{ns1}}}
\newcommand{\tnsr} {{\mathtt{ns2}}}

\newcommand{\lap} {{\lambda_\tp}}
\newcommand{\lan} {{\lambda_\tn}}

\newcommand{\ellps} {{\ell_\tps}}

\newcommand{\ellpu} {{\ell_\tpu}}
\newcommand{\ellnu} {{\ell_\tnu}}
\newcommand{\ellnsf} {{\ell_\tnsf}}
\newcommand{\ellnsr} {{\ell_\tnsr}}

\newcommand{\ellp} {{\ell_\tp}}
\newcommand{\elln} {{\ell_\tn}}

\newcommand{\mtrx}[4] {\begin{pmatrix}#1&#2\\#3&#4\end{pmatrix}}

\DeclareMathOperator{\poly}{poly}

\newcommand{\cA} {\ensuremath{\mathcal A}}
\newcommand{\cC} {\ensuremath{\mathcal C}}
\newcommand{\cD} {\ensuremath{\mathcal D}}

\newcommand{\cF} {\ensuremath{\mathcal F}}

\newcommand{\cH} {\ensuremath{\mathcal H}}
\newcommand{\cI} {\ensuremath{\mathcal I}}
\newcommand{\cL} {\ensuremath{\mathcal L}}
\newcommand{\cN} {\ensuremath{\mathcal N}}
\newcommand{\cP} {\ensuremath{\mathcal P}}
\newcommand{\cS} {\ensuremath{\mathcal S}}
\newcommand{\cU} {\ensuremath{\mathcal U}}

\newcommand{\Fnm}{\cF_{n,m}}

\newcommand{\Cnd}{\cC_{n,\mbf d}}
\newcommand{\Chnd}{\widehat\cC_{n,\db}}
\newcommand{\Cndh}{\cC_{n,\dbh}}

\newcommand{\xib}{{\mbf\xi}}

\newcommand{\ch}{{\widehat c}}
\newcommand{\cbh}{{\widehat{\mbf c}}}

\newcommand{\db}{{\mbf d}}
\newcommand{\dt}{{\widetilde d}}
\renewcommand{\dh}{{\widehat d}}
\newcommand{\dbh}{{\widehat{\mbf d}}}

\newcommand{\deltab}{{\mbf\delta}}
\newcommand{\deltah}{{\widehat\delta}}
\newcommand{\deltabh}{{\widehat{\mbf\delta}}}

\newcommand{\gammah}{{\widehat\gamma}}
\newcommand{\gammabh}{{\widehat{\mbf\gamma}}}

\newcommand{\form}{\phi}
\newcommand{\conf}{\varphi}
\newcommand{\conft}{{\widetilde\conf}}
\newcommand{\confh}{{\widehat\conf}}

\newcommand{\SAT} {\text{\sc Sat}}
\newcommand{\SIMPLE} {\text{\sc Simple}}
\newcommand{\cnf} {$3$-CNF}

\title{A new upper bound for 3-SAT\thanks{Partially
supported by the  and the Spanish CYCIT: TIN2007-66523
(FORMALISM). The first and second authors were also 
partially supported by \emph{La distinci\'o per a
la promoci\'{o} de la recerca} de la 
Generalitat de Catalunya, 2002.}}
\author{J.~D\' \i az$^1$  \qquad L.~Kirousis$^2$
\qquad D.~Mitsche$^1$ \qquad X.~P\' erez-Gim\' enez$^1$ 
\smallskip \\
{\small
$^1$Llenguatges i Sistemes Inform\`{a}tics, UPC, 08034
Barcelona }\\
{\small $^2$ Computer Engineering and Informatics, 
University of Patras, GR-26504 Rio, Patras}  \\
{\small\tt \{diaz,xperez,dmitsche\}@lsi.upc.edu, 
kirousis@ceid.upatras.gr}}

\date{\today}

\begin{document}
\maketitle
\begin{abstract}
We show that a randomly chosen 
$3$-CNF formula over $n$ variables with clauses-to-variables 
ratio at least $4.4898$ is, as  $n$ grows large, asymptotically almost surely unsatisfiable. 
The previous best such bound, 
due to Dubois in 1999, was $4.506$. 
The first such bound, independently 
discovered by many groups of researchers since 1983, 
was $5.19$. Several decreasing  values between 
$5.19$ and $4.506$ were published in the years between. 
The probabilistic techniques we use for the proof are, we believe, of independent interest.
\end{abstract}
\section{Introduction}

\bigskip\noindent
Satisfiability of  Boolean formulas  is 
a problem universally believed to be hard. 
Determining  the source of this hardness will lead, 
as is often stressed, to applications in  domains  
even outside the realm of  mathematics or computer science; 
moreover, and perhaps more importantly, 
it will   enhance our understanding of 
the foundations of computing.

In the beginning of the 90's several groups 
of experimentalists chose to examine 
the source of this hardness from the following viewpoint: 
consider a {\em random} $3$-CNF formula with  a given 
clauses-to-variables ratio, which is known as the {\em density} 
of the formula. What is the probability of it being satisfied 
and how does this probability depend on the density?  
Their simulation results led to the conclusion  that 
if the density is fixed and below a number  
approximately equal to $4.25$, then for large $n$,  
a randomly chosen  formula 
is  almost always  satisfiable, 
whereas if the density is fixed and above $4.25$, 
a randomly chosen formula is, for large $n$, 
almost always  unsatisfiable. More importantly, 
around $4.25$ the complexity of several well known 
complete algorithms for checking satisfiability 
reaches a steep peak 
(see e.g. \cite{kirkpatrick,mitchell92,MonZech}). 
So, in a certain sense, $4.25$ is  the point where 
from an empirical, statistical  viewpoint, 
the ``hard" instances of SAT  are to  be found. 
Similar results were obtained for other 
combinatorial problems, and also for $k$-SAT for values of $k>3$.

These experimental results were followed by an intense 
activity to provide ``rigorous results"    
(the expression often used in this context 
to refer  to theorems). 
Perhaps the most important advance is due to 
Friedgut: in ~\cite{friedgut99}
he proved that there is a sequence of reals 
$(c_n)_n$ such that for any $\epsilon>0$ 
the probability of a randomly chosen 3-CNF-formula with density  
$c_n -\epsilon$ being satisfiable approaches $1$ (as $n\rightarrow \infty$), 
whereas for density  $c_n +\epsilon$, it approaches $0$. 
Intuitively, this means that the transition from 
satisfiability to unsatisfiability is sharp, 
however it is still not known 
if $(c_n)_n$ converges.

Despite the fact that  the convergence of $(c_n)_n$  
is still an open problem, 
increasingly improved upper and lower 
bounds on its terms have been computed 
in a rigorous way by many groups of researchers. 
The currently best lower bound is 
$3.52$~\cite{kkl06,Coppersmith}.

With respect to upper bounds, 
which is the subject of this work, 
the progress was slower but better, 
in the sense that the experimentally established 
threshold is more closely bounded from above, 
rather than from below. 
A na\"{\i}ve application of the first moment 
method yields an upper bound of 5.191 
(see e.g., \cite{franco83}). 
An important advance was made in~\cite{kamath95}, where 
the upper bound was improved to $4.76$. 
In the sequel, the work of  several groups of researchers, 
based on more refined variants  of the  first moment method, 
culminated in the value of 
$4.571$~\cite{DuBo97,kkks98} 
(see the nice surveys~\cite{ksz05,dubsurv} for a complete
sequence of the events). 
The core idea in these works was to use 
the first moment method by computing the expected 
number of not all satisfying truth assignments, 
but only of those among them that  are local maxima 
in the sense of a lexicographic ordering, within 
a degree of locality determined by the 
Hamming distance between  truth assignments 
(considered as binary sequences). 
For degree of locality $1$, this amounts 
to computing the expected number of 
satisfying  assignments that become 
unsatisfying assignments by flipping any of their ``false" 
values (value $0$) to ``true" (value $1$). 
Such assignments are sometimes referred to as 
{\em single-flip} satisfying assignments.

The next big step was taken by 
Dubois et al.~\cite{DuBoMan00}, who showed that  $4.506$ 
is an upper bound. Instead of  considering further 
variations  of satisfiability, 
they limited the domain of computations 
to  formulas that have a 
typical {\em syntactic} characteristic. 
Namely, they considered formulas where 
the cardinality of variables with given numbers 
of occurrences as positive and negative  literals, 
respectively,  approaches a  two dimensional 
Poisson distribution. Asymptotically almost 
all formulas have this typical property 
(we say that such formulas have a Poisson 
{\em 2D degree sequence}). It turns out that the 
expectation of the number of single-flip 
satisfying  assignments is exponentially 
reduced when computed for such formulas. 
To get the afore mentioned upper bound, 
Dubois et al. further limited the domain 
of computations to formulas that are 
{\em positively unbalanced}, i.e.\ formulas where every variable 
has  at least as many occurrences as a positive 
literal as it has as a negated one.

A completely different direction  was recently 
taken in~\cite{maneva08}. Their work was   
motivated by  results on the geometry  of  
satisfying assignments, and especially 
the way they form clusters (components where 
one can move from one satisfying  assignment 
to another by hops of small Hamming distance). 
Most of these results  were originally based on 
analytical, but non-rigorous, techniques of  
Statistical Physics; lately however important 
rigorous advances were made~\cite{achlioptas06,mezard02}.
The value of the upper bound obtained by Maneva and 
Sinclair (see~\cite{maneva08}) was $4.453$, far below any other 
upper bound presently known (including the one in this paper). 
However it was proved assuming a conjecture on 
the geometry of the satisfying truth assignments 
which is presently proved only for $k$-SAT for $k \geq 8$ 
in~\cite{achlioptas06}.

In this paper, we show that $4.4898$ 
is an upper bound. Our approach builds upon previous work. 
It makes use (i) of single-flip satisfying truth assignments, 
(ii) of formulas with a Poisson 2D degree sequence and 
(iii) of positively unbalanced formulas.

 We add to these previously known techniques 
two novel elements that further reduce 
the expectation computed. Our approach is rigorous: 
although we  make use of computer programs, 
the outputs we use are formally justified.   
What is interesting is not that much the numerical 
value we get, although it  constitutes a 
further improvement to a long series of results. 
The main interest lies, we believe,  
on one hand in the new  techniques 
themselves and on the other in the fact that 
putting together so many disparate techniques 
necessitates a delicately balanced proof structure.

First, we start by recursively eliminating 
one-by-one the occurrences of pure literals 
from the random formula, until we get  its 
{\em impure core}, i.e.\ the largest sub-formula 
with no pure literals (a pure literal is one 
that has at least one occurrence in the formula 
but whose negation has none). Obviously this 
elimination has no effect on the satisfiability 
of the formula. Since we consider random 
formulas with a given 2D degree sequence, we first have  
to determine what is the 2D degree sequence of the impure core. 
For this, we use the differential equation 
method~\cite{wormald99}. The setting of the 
differential equations is more conveniently 
carried out in the so called {\em configuration} model, 
where the random formula is constructed by starting 
with  as many labelled copies of each literal  
as its occurrences and then by considering 
random 3D matchings of these copies. The  matchings   
define the clauses. The change of models from the standard one to the configuration model with a Poisson 2D degree sequence is formalized in Lemma~\ref{lem:Fnm}. We also  take 
care of the fact that the configuration model 
allows formulas with (i) multiple clauses and 
(ii) multiple occurrences of the same variable in a clause,  
whereas we are interested in simple formulas, i.e.\ 
formulas where neither (i) nor (ii)  holds. 
For our purposes, it is enough  to bound 
from below the probability of getting a 
simple formula in the configuration model 
by $e^{-\Theta (n^{1/3} \log n)}$,  see 
Lemma~\ref{lem:simple}. 
The differential equations are then  
analytically solved, and we thus obtain  
the 2D degree sequence of the core,  
see Lemma~\ref{lem:dem}.

Second, we require that not only the 2D 
degree sequence is Poisson, but also 
that the numbers of clauses with none, 
one, two and three positive literals, respectively, are close to the expected numbers. Notice that these expected numbers have to reflect the fact that we consider positively unbalanced formulas.
This is formalized in Lemma~\ref{lem:clause}. 

The expectation of the number of satisfying assignments, in the framework determined by all the restrictions above,  
is computed  in Lemma~\ref{lem:hard}. 
This expectation turns out to be given  
by a sum of polynomially many terms of  
functions that are exponential in $n$. 
We estimate this sum by its maximum term, 
using a standard technique. 
However in this case, finding the maximum term 
entails maximizing a function of many variables whose number depends on $n$. 
To avoid  a maximization that depends on $n$ 
we prove a truncation result which allows us 
to consider formulas that have a Poisson 2D 
degree sequence only for {\em light} variables, i.e.\  
variables whose  number of occurrences, 
either as  positive or negated literals,
is at most  a constant independent of $n$.

Then 
we carry out the maximization. 
The technique we use is the standard one 
by Lagrange multipliers. We get a 
complex $3\times3$ system which can be solved numerically. We formally prove that the system does not maximize on the boundary of the system and we make a sweep over the domain which confirms the results of the numerical solution.

Due to lack of space, the proofs of some lemmata are deferred to the appendix of the paper. 

As usual asymptotically almost surely (a.a.s.) will mean with probability tending to $1$  as $n\rightarrow \infty$. All asymptotic expressions as $1-o(1)$ are always with respect to $n$.
Our main result in the paper is the following:
\begin{thm}\label{thm:main}
Let $\gamma=4.4898$ and $m=\lfloor \gamma n\rfloor$.
A random \cnf\ formula in $\Fnm$ (i.e.\ with $n$ variables and $m$ clauses, no repetition of clauses and no repetition of 
variables in a clause) is not satisfiable a.a.s.
\end{thm}

\section{Background and Technical highlights.}

Given a set of $n$ Boolean variables, and 
let $m=\lfloor \gamma n\rfloor$. 
Let $\Fnm$ be the set of \cnf\ formulas with $n$ variables 
and $m$ clauses, where repetition of clauses or repetition of 
variables in a clause is not allowed. 
We also denote by $\Fnm$ the probability space of formulas in $\Fnm$ 
drawn with uniform probability. Throughout the paper, we fix the value 
$\gamma=4.4898$ and  prove for that value a random \cnf\ formula is not satisfiable with
high probability.

Throughout the paper, \emph{scaled} will always mean divided by $n$, and a 
\emph{scaled natural} will be a member of $\frac1n\nat$. 
Given formula $\form\in\Fnm$, we define the following parameters which depend on $\form$:
For any $i,j\in\nat$, let $d_{i,j}$ be the scaled number of variables with $i$ 
positive occurrences and $j$ negative occurrences in $\form$. Then,
\begin{equation}
\sum_{i,j\in\nat} d_{i,j}=1.
\label{eq:sumto1}
\end{equation}
The sequence $\db =(d_{i,j})_{i,j\in\nat}$ is called the
\emph{degree sequence} of $\form$. The scaled number of clauses of $\form$ is denoted by $c$, and can be expressed by 
\begin{equation}
c(\db)=\frac13\sum_{i,j\in\nat} (i+j) d_{i,j}.
\label{eq:sumtoc}
\end{equation}
Note that if $\form\in\Fnm$, then $c$ must additionally satisfy $c=\lfloor \gamma n\rfloor/n$.

Given $\epsilon_1>0$ and any sequence $\xib=(\xi_{i,j})_{i,j\in\nat}$ 
of nonnegative reals with $\sum_{i,j\in\nat} \xi_{i,j}=1$, define 
\begin{align*}
\cN(n,\xib,\epsilon_1)=
\Big\{ & \db=(d_{i,j})_{i,j\in\nat} \st
\sum_{i,j\in\nat} d_{i,j}=1, \quad
\frac{n}{3}\sum_{i,j\in\nat} (i+j) d_{i,j}\in\nat,
\quad\forall i,j\in\nat \quad d_{i,j}n\in\nat,
\\
&\text{and}\quad |d_{i,j}-\xi_{i,j}|\le\epsilon_1,
\quad\text{and if $i>n^{1/6}$ or $j>n^{1/6}$ then $d_{i,j}=0$} \Big\}.
\end{align*}
Intuitively $\cN(n,\xib,\epsilon_1)$ can be interpreted as 
the set of degree sequences $\db$ which are close to the ideal sequence 
$\xib$, 
which in general is not a degree sequence since its entries $\xi_{i,j}$ need not be 
scaled naturals. However, if $n$ is large enough, then 
$\cN(n,\xib,\epsilon_1)\ne\emptyset$.  Now we consider the 2D Poisson ideal sequence $\deltab$ defined by
$\delta_{i,j}=e^{-3\gamma}(3\gamma/2)^{i+j}/(i!j!)$.
The following lemma reflects the fact that almost all $\form\in\Fnm$ have a degree sequence $\db$ which is close to $\deltab$. A proof of an analogous result can be found in~\cite{DuBoMan00} (the main difference is that we have to apply Markov's inequality in the end to show that there are no variables of degree $n^{1/6}$ or larger).
\begin{lem}\label{lem:Fnm}
Let $\db$ be the degree sequence of a random $\form\in\Fnm$. 
For any $\epsilon_1>0$, we have that $\pr_{\Fnm} (\db\in\cN(n,\deltab,\epsilon_1))=1-o(1)$.
\end{lem}
%
%

Given a fixed degree sequence $\db=(d_{i,j})_{i,j\in\nat}$ satisfying~\eqref{eq:sumto1} 
and such that $c=c(\db)$ defined by~\eqref{eq:sumtoc} is also a scaled natural, 
we wish to generate \cnf\ formulas with that particular degree sequence $\db$. 
A natural approach to this is to use the \emph{configuration model}. 
A \emph{configuration} $\conf$ with degree sequence $\db=(d_{i,j})_{i,j\in\nat}$ 
is constructed as follows: consider $n$ variables and the corresponding 
$2n$ literals $x_1,\bar x_1\ldots,x_n,\bar x_n$; each literal 
has a certain number of distinct labelled \emph{copies} in a way that 
the scaled number of variables with $i$ positive copies and $j$ 
negative copies is $d_{i,j}$; then partition the set of copies 
into sets of size $3$, which we the call clauses of $\conf$.
Let $\Cnd$ be the set of all configurations with degree sequence $\db$, 
and we also denote by $\Cnd$ the probability space on the set $\Cnd$ 
with the uniform distribution.

A \cnf\ \emph{multi-formula} is a formula with possible repetition of variables 
in one clause and/or 
possible repetition of clauses. A \emph{simple formula} is a formula in $\Fnm$.
Let $\pi$ be the  projection from $\Cnd$  to \cnf\ multi-formulas  obtained 
by unlabelling the copies of each literal. 
A  configuration $\conf\in\Cnd$ is satisfiable if $\form=\pi(\conf)$ 
is satisfiable. A configuration $\conf\in\Cnd$ 
is \emph{simple} iff $\form=\pi(\conf)$ is a simple formula, i.e.\ does not have repetition of variables or clauses. 
Notice that the number of anti-images of a simple formula $\form$ with degree sequence $\db$ under $\pi$  
does not depend on the particular choice of $\form$. 
Hence,
\begin{equation}
\pr_{\Fnm}(\form \text{~is~} \SAT\mid\db) = \pr_{\Cnd}(\conf \text{~is~} \SAT\mid\SIMPLE).
\label{eq:FtoC}
\end{equation}
%

We need a lower bound on the probability that a configuration is simple. 
The following result gives a weak bound which is enough for our purposes. 
\begin{lem}\label{lem:simple}

Let $\epsilon_1>0$ and $\db\in\cN(n,\deltab,\epsilon_1)$. Then 
\[
\pr_{\Cnd}(\SIMPLE) \geq e^{-\Theta(n^{1/3}\log n)},
\]
where the  $e^{-\Theta(n^{1/3}\log n)}$ bound is uniform for all 
$\db\in\cN(n,\deltab,\epsilon_1)$.
\end{lem}
Given $\conf\in\Cnd$, a \emph{pure} 
variable of $\conf$ is a variable 
which has a non-zero number of occurrences 
which are either all syntactically positive or 
all syntactically negative. The only literal 
occurring in $\conf$ and all its copies 
are also called \emph{pure}.
If $\conf$ is satisfiable 
and $x$ is a pure variable of $\conf$, 
then there exists some satisfying truth 
assignment of $\conf$ which satisfies all 
copies of $x$ in $\conf$. Hence, in order 
to study the satisfiability of a 
$\conf\in\Cnd$, we can 
satisfy each pure variable in $\conf$ and 
remove all clauses containing a copy of that variable. 
For each $\conf\in\Cnd$, let $\conft$ be the configuration 
obtained by greedily 
removing all pure variables
and their corresponding clauses 
from $\conf$. 
This $\conft$ is 
independent of the particular elimination order 
of pure literals and is called the \emph{impure core} of $\conf$. In fact, in our analysis we will eliminate only one clause containing one copy of a pure literal at a time (the $\conft$ obtained still remains the same).
\remove{This fact can be proved in a completely analogous 
way to the standard argument that shows 
that the obtention of the $k$-core of a 
graph does not depend on how we remove 
the vertices of degree less than $k$.}
Note that $\conf$ is satisfiable 
iff $\conft$ is satisfiable. \remove{since 
we are allowed to choose a truth assignment 
which satisfies all literals  
iteratively removed during the process 
for obtaining $\conft$.} Moreover, 
if $\conf$ is simple then $\conft$ is 
also simple (but the converse is not necessarily true).

Furthermore, let $\confh$ be the configuration obtained from
 $\conft$ by positively unbalancing all variables, i.e.\
switching the syntactic sign of those 
variables having initially more 
negative than positive occurrences in $\conft$.
Let $\Chnd$ denote the probability space 
of configurations $\confh$, where $\conf$ 
was  chosen from $\Cnd$ with 
uniform probability. Note that 
the probability distribution in $\Chnd$ 
is not necessarily uniform. Since 
the simplicity and the satisfiability 
of a configuration are not affected by positively unbalancing the variables, we have
\begin{equation}
\pr_{\Cnd}(\conf \text{~is~} \SAT\wedge\SIMPLE) 
\le \pr_{\Chnd}(\confh \text{~is~} \SAT\wedge\SIMPLE).
\label{eq:whycore}
\end{equation}
Let the random variable $\dbh$ be the degree sequence 
of a random configuration in $\Chnd$. 
We prove in the following result that 
if the original 
$\db$ is close to the ideal sequence $\deltab$, 
then with high probability $\dbh$ must 
be close to the ideal sequence 
$\deltabh=(\deltah_{i,j})_{i,j\in\nat}$
defined by
\[
\deltah_{i,j} = \begin{cases}
2e^{-3\gamma b}\frac{(3\gamma b/2)^{i+j}}{i!j!},
& \text{if } i > j,
\\
e^{-3\gamma b}\frac{(3\gamma b/2)^{i+j}}{i!j!},
& \text{if } i=j,
\\
0,& \text{if } i < j, 
\end{cases}
\]
where $b=(1-t_\cD/\gamma)^{2/3}$ and $t_\cD$ is the scaled number of steps 
in the pure literal elimination algorithm.
%
\begin{prop}\label{lem:dem}
Given $\epsilon_2>0$, there exists $\epsilon_1>0$ and $0<\beta<1$ such 
that for any $\db\in\cN(n,\deltab,\epsilon_1)$
\[
\pr_{\Chnd}\Big(\dbh\in\cN(n,\deltabh,\epsilon_2)\Big) = 1-O(\beta^{n^{1/2}}).
\]
Moreover, for each $\dbh\in\cN(n,\deltabh,\epsilon_2)$, 
the probability space $\Chnd$ conditional upon having 
degree sequence $\dbh$ has the uniform distribution (i.e.\ $\Chnd$ conditional upon a fixed $\dbh$ behaves exactly as $\Cndh$).
\end{prop}
Let $\dbh\in\cN(n,\deltabh,\epsilon_2)$. Then, each $\conf\in\Cndh$ has a scaled number of clauses of $\ch=c(\dbh)$ (see~\eqref{eq:sumtoc}). Moreover, let $\ellp$ and $\elln$ be the scaled number of copies in $\conf$ of positive and of negative literals respectively. Then
\begin{equation}\label{eq:lpln_d}
\ellp(\dbh)=\sum_{i,j\in\nat} i \dh_{i,j},\qquad \elln(\dbh)=\sum_{i,j\in\nat} j \dh_{i,j}.
\end{equation}
Given any fixed $\conf\in\Cndh$ and for $k\in\{0,\ldots,3\}$, let $\ch_k$ 
be the scaled number of clauses in $\conf$ containing exactly $k$ positive copies 
(clauses of \emph{syntactic type k}). We call $\cbh=(\ch_0,\ldots,\ch_3)$ the 
\emph{clause-type sequence} of $\conf$. By definition 
\begin{equation}\label{eq:lpln_c}
\ch_1+2\ch_2+3\ch_3 = \ellp, \qquad
3\ch_0+2\ch_1+\ch_2 = \elln,
\end{equation}
and by adding the equations in~\eqref{eq:lpln_c}, $\ch_0+\cdots+\ch_3=\ch$. The  
$\ch_0,\ldots,\ch_3$ are random variables in $\Cndh$, but the next result 
shows that if $\dbh$ is close enough to $\deltabh$, then $\ch_0,\cdots,\ch_3$ as well as their sum $\ch_0+\cdots+\ch_3=\ch$ are concentrated with high probability. In order to see this, we need to define $\gammah=c(\deltabh)$, $\lap=\ellp(\deltabh)$ and $\lan=\elln(\deltabh)$ 
(see~\eqref{eq:sumtoc} and~\eqref{eq:lpln_d}),
which can be interpreted as the limit of $\ch$, $\ellp$ and $\elln$ respectively when $\dbh$ approaches $\deltabh$. In terms of these numbers, we thus define for all $k\in\{0,\ldots,3\}$
\begin{equation}\label{eq:gammas}
\gammah_k = \binom{3}{k}\frac{\lap^k \lan^{3-k}}{(\lap+\lan)^3} \, \gammah
\end{equation}
and also $\gammabh=(\gammah_0,\ldots,\gammah_3)$. Then we have
$\gammah_1+2\gammah_2+3\gammah_3 = \lap$, 
$3\gammah_0+2\gammah_1+\gammah_2 = \lan$ and
$\gammah_0+\gammah_1+\gammah_2+\gammah_3 = \gammah$.

The next result shows that when $\dbh$ is close enough to $\deltabh$, 
then each $\ch_k$ is close to the corresponding $\gammah_k$. Indeed, given $\epsilon>0$ and for any  $\dbh\in\cN(n,\deltabh,\epsilon_2)$, 
let $\Cndh^\epsilon$ be the set of all $\conf\in\Cndh$ such that for $k\in\{0,\ldots,3\}$,
$|\ch_k-\gammah_k|\le\epsilon$.  We also denote by $\Cndh^\epsilon$ the corresponding uniform probability space.
\begin{lem}\label{lem:clause}
Given $\epsilon>0$, there is $\epsilon_2>0$ and $0<\beta<1$ such that for any 
$\dbh\in\cN(n,\deltabh,\epsilon_2)$,
\[
\pr_{\Cndh}(\Cndh^\epsilon) = 1 - O(\beta^n).
\]
\end{lem}
All the previous lemmata establish a connection 
between the uniform probability spaces $\Fnm$ 
and $\Cndh^\epsilon$. In order to prove Theorem~\ref{thm:main}, 
it remains to bound the probability that a configuration 
$\conf\in\Cndh^\epsilon$ is simple and satisfiable, as it is done in the following result.
\begin{lem}\label{lem:hard} 
There exists $\epsilon>0$ and $0<\beta<1$ such that for any $\dbh\in\cN(n,\deltabh,\epsilon)$,
\[
\pr_{\Cndh^\epsilon}(\SAT\wedge\SIMPLE) = O(\beta^n).
\]
\end{lem}
The proof of Lemma~\ref{lem:hard} is given in Section~\ref{sec:lemhard} below. Finally, we can complete the proof of our main result.

\begin{proof}[Proof of Theorem~\ref{thm:main}]
Choose $\epsilon,\epsilon_2,\epsilon_1>0$ and $0<\beta<1$ which satisfy the statements of Lemmata~\ref{lem:dem}, \ref{lem:clause} and~\ref{lem:hard}. We can assume that $\epsilon_1<\epsilon_2<\epsilon$.
Let $\cN_1 = \cN(n,\deltabh,\epsilon_2) \subset \cN(n,\deltabh,\epsilon)$.
Using Lemmata~\ref{lem:clause} and \ref{lem:hard}, for any $\dbh\in\cN_1$
\begin{equation}
\pr_{\Cndh}(\SAT\wedge\SIMPLE) \le \pr_{\Cndh^\epsilon}(\SAT\wedge\SIMPLE) + O(\beta^n) = O(\beta^n).
\label{eq:mainthm1}
\end{equation}
Let $\cN_2 = \cN(n,\mbf\delta,\epsilon_1)$, and consider any given $\db\in\cN_2$. 
Then {from}~\eqref{eq:whycore}, \eqref{eq:mainthm1} and Lemma~\ref{lem:dem}, 
\begin{align}
\pr_{\Cnd}(\SAT\wedge\SIMPLE) &\le \pr_{\Chnd}(\SAT\wedge\SIMPLE)
\notag\\
&= \pr_{\Chnd}(\dbh\notin\cN_1) + \sum_{\dbh\in\cN_1} \pr_{\Cndh}(\SAT\wedge\SIMPLE) \pr_{\Chnd}(\dbh)
\notag\\
&\le O(\beta^{n^{1/2}}) + O(\beta^n) \pr_{\Chnd}(\dbh\in\cN_1)
= O(\beta^{n^{1/2}}),
\label{eq:mainthm2}
\end{align}
where  the events $(\dbh)$ and $(\dbh\in\cN_1)$ on the probability space $\Chnd$ denote  
that the degree sequence is $\dbh$ and  the degree sequence is in $\cN_1$.
Finally, from~\eqref{eq:FtoC},  Lemma~\ref{lem:Fnm}, Lemma~\ref{lem:simple} and~\eqref{eq:mainthm2} we get
\begin{align*}
\pr_{\Fnm}(\SAT) &= \pr_{\Fnm}(\db\notin\cN_2) + \sum_{\db\in\cN_2} \pr_{\Cnd}(\SAT\mid\SIMPLE) \pr_{\Fnm}(\db)
\\
&= o(1) + \sum_{\db\in\cN_2} \frac{\pr_{\Cnd}(\SAT\wedge\SIMPLE)}{\pr_{\Cnd}(\SIMPLE)}\pr_{\Fnm}(\db)
\\
&\le o(1) + e^{\Theta(n^{1/3}\log n)} \sum_{\db\in\cN_2} \pr_{\Cnd}(\SAT\wedge\SIMPLE)\pr_{\Fnm}(\db)
\\
&\le o(1) + e^{\Theta(n^{1/3}\log n)} O(\beta^{n^{1/2}}) \pr_{\Fnm}(\db\in\cN_2)= o(1),
\end{align*}
where  $(\db)$ and $(\db\in\cN_2)$ denote respectively the event
that a random formula $\form\in\Fnm$ has degree sequence $\db$ and the event that the 
degree sequence of $\form$ belongs to $\cN_2$.
\end{proof}
\section{Proof of Lemma~\ref{lem:hard}}\label{sec:lemhard}
Let $\cN(n,\deltabh,\gammabh,\epsilon)$ be the set of tuples $(\dbh,\cbh)$ such that $\dbh\in\cN(n,\deltabh,\epsilon)$ and $\cbh=(\ch_k)_{0\le k\le3}$ is a tuple of scaled naturals satisfying~\eqref{eq:lpln_c} (recall also from~\eqref{eq:lpln_d} the definition of $\ellp$ and $\elln$), and moreover $|\ch_k-\gammah_k|\le\epsilon$. For each $(\dbh,\cbh)\in\cN(n,\deltabh,\gammabh,\epsilon)$, we define $\cC_{n,\dbh,\cbh}$ to be the uniform probability space of all configurations with degree sequence $\dbh$ and clause-type sequence $\cbh$.
In order to prove the lemma, it suffices to show that for any $(\dbh,\cbh)\in\cN(n,\deltabh,\gammabh,\epsilon)$ we have
\[
\pr_{\cC_{n,\dbh,\cbh}}(\SAT\wedge\SIMPLE) = O(\beta^n).
\]
Hence, we consider $\dbh$, $\cbh$ and the probability space $\cC_{n,\dbh,\cbh}$ to be fixed throughout this section, and we try to find a suitable bound for $\pr(\SAT\wedge\SIMPLE)$.

We need some definitions.
Let us fix any given configuration $\conf\in\cC_{n,\dbh,\cbh}$. A \emph{light} variable of $\conf$ is a variable with $i\le M$ positive occurrences and $j\le M$ negative occurrences in $\conf$ (we will use here and in the numerical calculations that follow the value $M=23$). The other variables are called \emph{heavy}. We consider a weaker notion of satisfiability in which heavy variables are treated as jokers and are always satisfied regardless of their sign in the formula and their assigned value.
Given a configuration $\conf\in\cC_{n,\dbh,\cbh}$ and a truth assignment $A$, we say that $A\models^\flat\conf$ iff each clause of $\conf$ contains at least one heavy variable or at least one satisfied occurrence of a light variable.
Let $\SAT^\flat$ be the set of configurations $\conf\in\cC_{n,\dbh,\cbh}$ for which there exists at least one truth assignment $A$ such that $A\models^\flat\conf$. Clearly, if $A\models\conf$, then also $A\models^\flat\conf$, and hence $\SAT\subset\SAT^\flat$. We still introduce a further restriction to satisfiability in a way similar to~\cite{kkks98} and~\cite{DuBo97}, in order to decrease the number of satisfying truth assignments of each configuration without altering the set of satisfiable configurations (at least without alterating this set for simple configurations). Given a configuration $\conf\in\cC_{n,\dbh,\cbh}$ and a truth assignment $A$, we say that $A\models^{\flat\prime}\conf$ iff $A\models^\flat\conf$ and moreover each light variable which is assigned the value zero by $A$ appears at least once as the only satisfied literal of a \emph{blocking clause} (i.e. a clause with one satisfied negative literal and two unsatisfied ones). Let $\SAT^{\flat\prime}$ be the set of configurations which are satisfiable according to this latter notion. Notice that if $\conf\in\SIMPLE$, then $\conf\in\SAT^{\flat\prime}$ iff $\conf\in\SAT^\flat$ (by an argument analogous to the one in~\cite{kkks98} and~\cite{DuBo97}). Therefore, we have
\[
\pr(\SAT\wedge\SIMPLE) \le \pr(\SAT^\flat\wedge\SIMPLE) = \pr(\SAT^{\flat\prime}\wedge\SIMPLE) \le \pr(\SAT^{\flat\prime}).
\]
Let $X$ be the random variable counting the number of satisfying truth assignments of a randomly chosen configuration $\conf\in\cC_{n,\dbh,\cbh}$ in the $\SAT^{\flat\prime}$ sense. We need to bound
\begin{equation}\label{eq:1mm}
\pr(\SAT^{\flat\prime}) = \pr(X>0) \le \ex X=\frac{|\{(\conf, A) \st \conf\in\cC_{n,\dbh,\cbh}, \; A \models^{\flat\prime} \conf \}|}{|\cC_{n,\dbh,\cbh}|}.
\end{equation}
In the following subsection, we obtain an exact but complicated expression for $\ex X$ by a counting argument, and then we give a simple asymptotic bound which depends on the maximization of a particular continuous function over a bounded polytope. The next subsection contains the maximization of that function.
%
%
%
%
\subsection{Asymptotic bound on $\ex X$}
First, we compute the denominator of the rightmost member in~\eqref{eq:1mm}.
\begin{align*}
|\cC_{n,\dbh,\cbh}| &=
\binom{n}{(\dh_{i,j}n)_{i,j}} \binom{\ellp n}{\ch_1n,2\ch_2n,3\ch_3n}
\binom{\elln n}{3\ch_0n,2\ch_1n,\ch_2n}
\frac{(3\ch_0n)!}{(\ch_0n)! 6^{\ch_0n}} \frac{(2\ch_1n)!}{2^{\ch_1n}} \frac{(2\ch_2n)!}{2^{\ch_2n}} \frac{(3\ch_3n)!}{(\ch_3n)! 6^{\ch_3n}}
\\
&= \frac{n!}{\prod_{i,j} (\dh_{i,j}n)!} \frac{(\ellp n)!(\elln n)!}
{2^{\ch n}3^{(\ch_0+\ch_3)n}(\ch_0n)!(\ch_1n)!(\ch_2n)!(\ch_3n)!}
\end{align*}
In order to deal with the numerator in~\eqref{eq:1mm}, we need some definitions. Let us consider any fixed $\conf\in\cC_{n,\dbh,\cbh}$ and any assignment $A$ such that $A\models^{\flat\prime}\conf$. We will classify the variables, the clauses and the copies of literals in $\conf$ into several types, and define parameters counting the scaled number of items of each type. Variables are classified according to their degree. A variable is said to have degree $(i,j)$ if it appears $i$ times positively and $j$ times negatively in $\conf$. Let $\cL$ and $\cH$, respectively,  be the set of possible degrees for light and heavy variables.
\[
\cL = \{(i,j)\in\nat^2 \st 0\le i,j\le M\}, \qquad
\cH = \{(i,j)\in\nat^2 \st i>M \text{ or } j> M\}.
\]
We also consider an extended notion of degree for light variables which are assigned $0$ by $A$. One of such variables has extended degree $(i,j,k)$ if it has degree $(i,j)$ and among its $j$ negative occurrences $k$ appear in a blocking clause (being the only satisfied literal of the clause). Let
\[
\cL' = \{(i,j,k)\in\nat^3 \st 0\le i\le M,\;1\le k\le j\le M\},
\]
be the set of possible extended degrees for these light $0$-variables.
For each $(i,j)\in\cL$, let $t_{i,j}$ be the scaled number of light variables assigned $1$ by $A$ with degree $(i,j)$ in $\conf$. For each $(i,j,k)\in\cL'$, let $f_{i,j,k}$ be the scaled number of light variables assigned $0$ by $A$ with extended degree $(i,j,k)$ in $\conf$. We must have
\begin{equation}\label{eq:tf's}
t_{i,j} + \sum_{k=1}^j f_{i,j,k} = \dh_{i,j}, \qquad \forall (i,j)\in\cL.
\end{equation}
On the other hand, we classify the copies of literals occurring in $\conf$ into five different types depending on their sign in $\conf$, their assignment by $A$ and whether they belong or not to a blocking clause. Each copy receives a label from the set
\[
\cS = \{\tps,\tnsf,\tnsr,\tpu,\tnu\},
\]
where the labels $\tps$, $\tpu$, $\tnsf$, $\tnsr$ and $\tnu$ denote positive-satisfied, positive-unsatisfied, negative-satisfied inside a blocking clause, negative-satisfied inside a non-blocking clause and negative-unsatisfied, respectively. It is useful to consider as well coarser classifications of the copies of literals in $\conf$ and thus we define the types $\tp$, $\tn$ and $\tns$ which correspond to positive, negative and negative-satisfied copies, respectively. Also, let
\[
\cS' = \{\tps,\tns,\tpu,\tnu\} \quad\text{and}\quad\cS'' = \{\tp,\tn\}.
\]
For each of the types $\sigma\in\cS\cup\cS'\cup\cS''$ that we defined, let $\ell_\sigma$ be the scaled number of copies of type $\sigma$. Note that $\ellp$ and $\elln$ were already defined (see~\eqref{eq:lpln_d} and~\eqref{eq:lpln_c}). Also, let $h_\sigma$ be the scaled number of copies of type $\sigma$ which come from heavy variables (recall that these copies are always satisfied by definition regardless of their sign). In view of
\[
h_\tps = \sum_\cH i\dh_{i,j}, \qquad h_\tns = \sum_\cH j\dh_{i,j}
\]
and of~\eqref{eq:lpln_d} and~\eqref{eq:lpln_c}, we observe that $\ellp$, $\elln$, $h_\tps$ and $h_\tns$ are constants which do not depend on the particular choice of $(\conf,A)$. The parameters $h_\tnsf$ and $h_\tnsr$ depend on the particular $(\conf,A)$ and satisfy
\begin{equation}\label{eq:h's}
h_\tnsf + h_\tnsr = h_\tns.
\end{equation}
The parameters $\ellps$, $\ellpu$, $\ellnsf$, $\ellnsr$ and $\ellnu$ also depend on $(\conf,A)$ and can be expressed as
\begin{align}
\ellps &= \sum_{\cL} it_{i,j} + h_\tps, &
\ellpu &= \sum_{\cL'} if_{i,j,k}, &
\ellnsf &= \sum_{\cL'} kf_{i,j,k} + h_\tnsf,
\notag\\
\ellnsr &= \sum_{\cL'} (j-k)f_{i,j,k} + h_\tnsr, &
\ellnu &= \sum_{\cL} jt_{i,j}.
\label{eq:tfh-ell's}
\end{align}
Finally, the clauses of $\conf$ are classified into $16$ extended types (not to be mistaken with the four syntactic types defined immediately before~\eqref{eq:lpln_c}). Each type is represented by a $2\times2$ matrix from the set
\[
\cA = \left\{ \alpha=\mtrx{\tps(\alpha)}{\tns(\alpha)}{\tpu(\alpha)}{\tnu(\alpha)} \in \mathbb N^4 \;:\; \sum_{\sigma\in\cS'} \sigma(\alpha) = 3, \; \tps(\alpha)+\tns(\alpha)>0 \right\}.
\]
A clause is said to be of extended type $\alpha=\mtrx{\tps(\alpha)}{\tns(\alpha)}{\tpu(\alpha)}{\tnu(\alpha)}$ if for each $\sigma\in\cS'$ the clause contains $\sigma(\alpha)$ copies of literals of type $\sigma$. Notice that all clauses of extended type $\alpha$ also contain the same number of copies of type $\sigma$ for all other $\sigma\in\cS\cup\cS''$ and thus we can define $\sigma(\alpha)$ to be this number.
For each $\alpha\in\cA$, let $c_\alpha$ be the scaled number of clauses of extended type $\alpha$ (while $\ch_k$, $0\le k\le3$ is the number of clauses of syntactic type $k$, i.e.\ with $k$ positive literals). We have
\begin{equation}\label{eq:c's}
\sum_{\substack{\alpha\in\cA\\ \tp(\alpha)=k}} c_\alpha = \ch_k.
\end{equation}
The parameters $\ellps$, $\ellpu$, $\ellnsf$, $\ellnsr$ and $\ellnu$ can also be expressed in terms of the $c_\alpha$ by
\begin{equation}\label{eq:c-ell's}
\ell_\sigma = \sum_{\alpha\in\cA} \sigma(\alpha)c_\alpha, \quad \forall\sigma\in\cS.
\end{equation}
We now consider the following equations:
\begin{equation}\label{eq:ell's}
\ellps+\ellpu = \ellp,
\qquad
\ellnsf+\ellnsr+\ellnu = \elln
\end{equation}
\begin{align}
\label{eq:tfh-ell's2}
\ellps &= \sum_{\cL} it_{i,j} + h_\tps,
&
\ellnsf &= \sum_{\cL'} kf_{i,j,k} + h_\tnsf,
&
\ellnsr &= \sum_{\cL'} (j-k)f_{i,j,k} + h_\tnsr
\\
\label{eq:c-ell's2}
\ellps &= \sum_{\alpha\in\cA} \tps(\alpha)c_\alpha,
&
\ellnsf &= \sum_{\alpha\in\cA} \tnsf(\alpha)c_\alpha,
&
\ellnsr &= \sum_{\alpha\in\cA} \tnsr(\alpha)c_\alpha
\end{align}
In view of~\eqref{eq:lpln_d} and~\eqref{eq:lpln_c}, the system of equations \{\eqref{eq:tf's}, \eqref{eq:h's}, \eqref{eq:tfh-ell's}, \eqref{eq:c's}, \eqref{eq:c-ell's}\} is equivalent to \{\eqref{eq:tf's}, \eqref{eq:h's}, \eqref{eq:c's}, \eqref{eq:ell's}, \eqref{eq:tfh-ell's2}, \eqref{eq:c-ell's2}\}.

So far we verified that the constraints \{\eqref{eq:tf's}, \eqref{eq:h's}, \eqref{eq:c's}, \eqref{eq:ell's}, \eqref{eq:tfh-ell's2}, \eqref{eq:c-ell's2}\} express necessary conditions for the parameters of any particular $(\conf,A)$, with $\conf\in\cC_{n,\dbh,\cbh}$ and $A\models^{\flat\prime}\conf$. Now we will see that they are also sufficient, in the sense that for each tuple of parameters satisfying the above-mentioned constraints we will be able to construct pairs $(\conf,A)$.

Let $\bar t = (t_{i,j})_{\cL}$, $\bar f = (f_{i,j,k})_{\cL'}$, $\bar h = (h_{\tnsf},h_{\tnsr})$, $\bar c = (c_\alpha)_{\alpha\in\cA}$, $\bar\ell = (\ell_\sigma)_{\sigma\in\cS}$ and
\[
K=|\cL|+|\cL'|+2+|\cA|+|\cS| = (M+1)^2(1+M/2) + 23.
\]
We define the bounded polytope $\cP(\dbh,\cbh)\subset\real^K$ as the set of tuples $\bar x = (\bar t, \bar f, \bar h, \bar c, \bar\ell)$ of non-negative reals satisfying \{\eqref{eq:tf's}, \eqref{eq:h's}, \eqref{eq:c's}, \eqref{eq:ell's}, \eqref{eq:tfh-ell's2}, \eqref{eq:c-ell's2}\}, and consider the following set of lattice points in $\cP(\dbh,\cbh)$:
\[
\cI(n,\dbh,\cbh) = \cP(\dbh,\cbh) \cap \left(\frac1n\nat\right)^{K}.
\]
For any tuple of parameters $\bar x\in\cI(n,\dbh,\cbh)$, we proceed to count the number of pairs $(\conf,A)$, with $\conf\in\cC_{n,\dbh,\cbh}$ and $A\models^{\flat\prime}\conf$, satisfying these parameters. We denote this number by $T(\bar x,n,\dbh,\cbh)$. In order to do the counting, we first consider the number of ways of assigning to each variable a truth value and a degree (and an extended degree as well for light $0$-variables). This gives
\[
2^{\sum_\cH \dh_{i,j}n}\binom{n}{(t_{i,j}n)_{\cL}, (f_{i,j,k}n)_{\cL'}, (\dh_{i,j}n)_{\cH}}.
\]
Notice that the only constraints required at this stage for the parameters are these in~\eqref{eq:tf's}. Then we have to choose for each light $0$-variable which $k$ of the $j$ negative satisfied copies contribute to blocking clauses, and also (taking into account~\eqref{eq:h's}) which $h_\tnsf$ of the $h_\tns$ negative copies of heavy variables contribute to blocking clauses:
\[
\left(\prod_{\cL'}\binom{j}{k}^{f_{i,j,k}n}\right) \binom{h_\tns n}{h_\tnsf n,h_\tnsr n}.
\]
So far, we have assigned labels in $\cS$ to all copies of literals in $\conf$. This construction is compatible with~\eqref{eq:tfh-ell's}. Now we have to decide which specific copies of each type $\sigma\in\cS$ will be used in the construction of clauses of each extended type $\alpha\in\cA$. The number of ways of doing that is
\[
\prod_{\sigma\in\cS} \binom{\ell_{\sigma}n}{(\sigma(\alpha)c_{\alpha}n)_{\alpha\in\cA}},
\]
where it is sufficient that the $\bar c$ parameters satisfy~\eqref{eq:c's} and that the $\bar\ell$ parameters are expressible as in~\eqref{eq:c-ell's}. Finally, the number of ways of constructing $c_\alpha$ clauses of each extended type $\alpha\in\cA$ is $\prod_{\alpha\in A} W(\alpha),$ where
\[
W(\alpha) = \frac{(w(\alpha)c_\alpha n)! (c_\alpha n)!^{2-w(\alpha)}}{(w(\alpha)!)^{c_\alpha n}} =
\left\{\begin{aligned}
&(c_\alpha n)!^2 & \text{if } w(\alpha)=1,
\\
& \frac{(2c_\alpha n)!}{2^{c_\alpha n}} & \text{if } w(\alpha)=2,
\\
& \frac{(3c_\alpha n)!}{(c_\alpha n)!6^{c_\alpha n}} & \text{if } w(\alpha)=3,
\end{aligned}\right.
\]
and $w(\alpha)$ is the number of $0$'s in the matrix $\alpha$. Putting everything together, we have that
\begin{align*}
T(\bar x,n,\dbh,\cbh) &= 2^{\sum_\cH \dh_{i,j}n}\binom{n}{(t_{i,j}n)_{\cL}, (f_{i,j,k}n)_{\cL'}, (\dh_{i,j}n)_{\cH}}
\\
&\quad\left(\prod_{\cL'}\binom{j}{k}^{f_{i,j,k}n}\right)
\binom{h_\tns n}{h_\tnsf n,h_\tnsr n}
\prod_{\sigma\in\cS} \binom{\ell_{\sigma}n}{(\sigma(\alpha)c_{\alpha}n)_{\alpha\in\cA}}
\prod_{\alpha\in A} W(\alpha).
\end{align*}
Hence
\[
\ex X = \frac{1}{|\cC_{n,\dbh,\cbh}|} \sum_{\bar x\in\cI(n,\dbh,\cbh)} T(\bar x,n,\dbh,\cbh).
\]
To characterize the asymptotic behaviour of $T(\bar x,n,\dbh,\cbh)/|\cC_{n,\dbh,\cbh}|$ with respect to $n$, we define
\[
F(\bar x) = \frac{\prod_{\sigma\in\cS}{\ell_{\sigma}}^{\ell_{\sigma}}}{\prod_{\cL}{t_{i,j}}^{t_{i,j}}  \prod_{\cL'}{\left(f_{i,j,k}/\binom{j}{k}\right)}^{f_{i,j,k}} {h_\tnsf}^{h_\tnsf} {h_\tnsr}^{h_\tnsr} \prod_{\alpha\in\cA} \big((w(\alpha)!/2)c_\alpha\big)^{c_\alpha}}
\]
and
\[
B(\dbh,\cbh) = 2^{\sum_{\cH}\dh_{i,j}} {h_\tns}^{h_\tns} \prod_{\cL}{\dh_{i,j}}^{\dh_{i,j}} \frac{3^{c_0+c_3}{c_0}^{c_0}{c_1}^{c_1}{c_2}^{c_2}{c_3}^{c_3}} {\ellp^\ellp \elln^\elln}.
\]
Then we use the following form of Stirling's inequality which holds for any $k\in\nat$:
\[
\sqrt{2\pi(k+1/8)}(k/e)^k\le k!\le\sqrt{2\pi(k+1/4)}(k/e)^k,
\]
and we obtain
\[
\frac{T(\bar x,n,\dbh,\cbh)}{|\cC_{n,\dbh,\cbh}|} \le
\poly_1(n) \big( B(\dbh,\cbh)  F(\bar x) \big)^n,
\]
where $\poly_1(n)$ is some fixed polynomial in $n$ which can be chosen to be independent of $\bar x$, $\dbh$ and $\cbh$ (as long as $\bar x\in\cI(n,\dbh,\cbh)$ and $(\dbh,\cbh)\in\cN(n,\deltabh,\gammabh,\epsilon)$). Moreover, since the size of $\cI(n,\dbh,\cbh)$ is also polynomial in $n$, we can write
\[
\ex X \le \poly_2(n) \left( B(\dbh,\cbh) \max_{\bar x\in\cI(n,\dbh,\cbh)}  F(\bar x) \right)^n
\le \poly_2(n) \left( B(\dbh,\cbh) \max_{\bar x\in\cP(n,\dbh,\cbh)} F(\bar x) \right)^n,
\]
for some other fixed polynomial $\poly_2(n)$.
By continuity, if we choose $\epsilon$ to be small enough, we can guarantee that
\begin{equation}\label{eq:bound_expectation}
\ex X \le \left( (1+ 10^{-7}) B \max_{\bar x\in\cP(n,\deltabh,\gammabh)} F(\bar x) \right)^n,
\end{equation}
where (recall the definition in~\eqref{eq:gammas})
\begin{align}
B = B(\deltabh,\gammabh) &= 2^{\sum_{\cH}\deltah_{i,j}} \left(\sum_{\cH}j\deltah_{i,j}\right)^{\sum_{\cH}j\deltah_{i,j}} \prod_{\cL}{\deltah_{i,j}}^{\deltah_{i,j}} \frac{3^{\gammah_0+\gammah_3}\, {\gammah_0}^{\gammah_0}{\gammah_1}^{\gammah_1}{\gammah_2}^{\gammah_2}{\gammah_3}^{\gammah_3}} {{\lap}^\lap{\lan}^\lan}
\notag\\
&= 2^{\sum_{\cH}\deltah_{i,j}} \left(\sum_{\cH}j\deltah_{i,j}\right)^{\sum_{\cH}j\deltah_{i,j}} \frac{\prod_{\cL}{\deltah_{i,j}}^{\deltah_{i,j}}}{(3\gammah)^{2\gammah}}.
\label{eq:B}
\end{align}

\subsection{Maximization of $\mbf{F(\bar x)}$}
We wish to maximize $F$ or equivalently $\log F$ over the domain $\cP(n,\deltabh,\gammabh)$. We first show the following lemma (for the proof see the appendix):
\begin{lem}\label{lem:boundary}
$F(\bar x)$ does not maximize on the boundary of $\cP(n,\deltabh,\gammabh)$.
\end{lem}
Since $\log F$ does not maximize on the boundary of its domain, the maximum must be attained at a critical point of $\log F$ in the interior of $\cP(n,\deltabh,\gammabh)$. We use the Lagrange multipliers technique and characterize each critical point of $\log F$ in terms of the solution of a $3\times 3$ system (the details are given in the appendix, Section~\ref{sec:lagrange}). The system is numerically solved with the help of Maple, which finds just one solution. We express the maximum of $F$ over $\cP(n,\deltabh,\gammabh)$ in terms of this solution, and multiply it by $B$ given in~\eqref{eq:B}, and from~\eqref{eq:bound_expectation} we obtain the bound
\begin{equation}\label{number:final}
\ex X \le \left( (1+ 10^{-7}) 0.9999998965 \right)^n,
\end{equation}
which concludes the proof of Lemma~\ref{lem:hard}, since $(1+ 10^{-7}) 0.9999998965<1$.

Note that the validity of our approach relies on the assumption that the solution of the $3\times3$ system found by Maple is unique, which implies that the critical point of $\log F$ we found is indeed the global maximum (if an alternative solution exists it could happen that at the corresponding critical point the function $F$ attains a value greater than the maximum obtained).

In order to be more certain about the correctness of~\eqref{number:final} we performed the following alternative experiment: Let $\cP_{\bar\ell}$ be the polytope obtained by restricting $\cP(n,\deltabh,\gammabh)$ to the coordinates $\ell_\tps,\ell_\tpu,\ell_\tnsf,\ell_\tnsr,\ell_\tnu$. Observe that this is a $3$-dimensional polytope in $\real^5$, since its elements are determined by the values of the coordinates $\ell_\tps,\ell_\tnsf,\ell_\tnsr$. We performed a sweep over this polytope by considering a grid of $100$ equispaced points in each of the three dimensions. For each of the $100^3$ fixed tuples of $(\ell_\tps,\ell_\tnsf,\ell_\tnsr)$ which correspond to the points on the grid, we determine the remaining two coordinates of $\cP_{\bar\ell}$, and maximize $\log F$ restricted to those fixed values of $\bar\ell$. Observe that in this case $\log F$ is strictly concave and thus has a unique maximum which can be efficiently found by any iterative Newton-like algorithm. We checked, again using Maple, that the value obtained for each fixed tuple of $\bar\ell$ is below the maximum in~\eqref{number:final}.

\newpage\begin{center}{\LARGE\bf Appendix}\end{center}\appendix
\section{Proof of Lemma~\ref{lem:simple}}
In order to prove the lemma, we give a lower bound on the number of simple configurations and then divide it by the total number of configurations.
The total number of configurations is
\[
\frac{(3cn)!}{6^{cn} (cn)!}.
\]
To obtain a lower bound on the number of simple configurations, we construct a subset of these. Consider the set of labelled copies of literals according to the degree sequence $\db$. Among those copies which correspond to variables with one unique occurrence, we select any $2 \lceil 9n^{1/3} \rceil$ (recall that in $\cN(n,\deltab,\epsilon_1)$ there is a linear number of such variables). Call this set $\cU_0$. Let $\cU$ be the set of all the remaining copies. We pick an arbitrary order of the copies in $\cU$ and put them into a list. At each step, we remove the first copy of the list, match it with two other suitable copies in the list (which are also removed) and repeat until the list contains exactly $\lceil 9n^{1/3} \rceil$ labelled copies of literals. Let $x$ be the variable corresponding to the first labelled copy of the list at some step of the procedure. If we want to avoid repetitions of variables in clauses then, when choosing the two other copies to be matched with that first one, we have to exclude the at most $2n^{1/6}$ other copies of $x$ (recall that in $\cN(n,\deltab,\epsilon_1)$ the degrees of all variables are at most $2n^{1/6}$). If we want to avoid multiple clauses, we have to exclude the at most $8n^{1/3}$ copies of the variables which are already in some clause together with some copy of $x$. Therefore, at most $\lceil 9n^{1/3} \rceil$ copies are excluded in total at each choice.
This gives rise to at least
\begin{equation}\label{eq:boundsimpleconf}
\frac{(3cn-3\lceil 9n^{1/3} \rceil)!}{6^{cn-\lceil 9n^{1/3} \rceil}(cn- \lceil 9n^{1/3}\rceil)!},
\end{equation}
different ways of grouping the copies in $\cU$ into clauses except for the $\lceil 9n^{1/3}  \rceil$ copies which remain in the list at the end of the procedure. Finally, each one of the $\lceil 9n^{1/3} \rceil$ remaining copies is matched arbitrarily with any two copies in $\cU_0$ to create a clause, and this completes the construction of the configurations. The choice of $\cU_0$ guarantees that all the configurations we obtained are simple. Hence, dividing~\eqref{eq:boundsimpleconf} by the total number of configurations gives a lower bound of $e^{-\Theta(n^{1/3}\log n)}$.

\section{Proof of Lemma~\ref{lem:dem}}
Given $\conf\in\Cnd$ with $\db\in\cN(n,\mbf\delta,\epsilon_1)$, we apply the following pure literal elimination algorithm to $\conf$: at each time step $t \geq 0$, one literal occurrence is chosen uniformly at random from all pure literal occurrences appearing in $\conf=\conf(t)$ at time $t$ (if the variable of the chosen literal occurrence is not yet assigned a value, it is set such that this literal occurrence is satisfied) and the clause containing this occurrence is eliminated. The algorithm stops when there is no more pure literal occurrence left. The resulting configuration is denoted by $\conft$. It is easy to see that $\conft$ is unique and independent of the order of elimination of pure literal occurrences. Moreover, at any time step $t$, the algorithm retains the conditional randomness of the configuration, i.e., conditional under having any (fixed) degree sequence $\db$ at time $t$, any configuration $\conf$ obeying this degree sequence and whose image $\form=\pi(\conf)$ is a simple formula, is equally likely to appear at time $t$ in this algorithm.
To analyze the expected degree distribution of $\conft$, we introduce the following variables: for any $0 \leq i,j \leq n$, $Y_{i,j}(t)$ denotes the random variable counting the number of variables with $i$ positive and $j$ negative occurrences at time $t$, and $Y_0(t)$ denotes the random variable counting the number of pure literal occurrences
at time $t$, i.e., $Y_0(t)=\sum_{i > 0}iY_{i,0}(t)+\sum_{j > 0} jY_{0,j}(t)$. Also define by $L(t)$ the random variable counting the total number of literal
occurrences at time $t$. Clearly, $L(t+1)-L(t)=-3$. Note also that
$Y_0(t)=L(t)-\sum_{i,j > 0}(i+j) Y_{i,j}(t)$.
Defining by $G(t)$ the sequence of $Y_0(t),Y_{1,1}(t),\ldots,Y_{n,n}(t)$, we have that for any $t$, for any $i,j > 0$ and conditional upon any values of $G(t)$ 
\begin{align*}
\ex\big(Y_{i,j}(t+1)-Y_{i,j}(t) \mid G(t)\big) &=
\frac{2}{L}\big((i+1)Y_{i+1,j}(t)+(j+1)Y_{i,j+1}(t)-(i+j)Y_{i,j}(t)\big)(1+o(1)),
\\
\ex\big(Y_0(t+1)-Y_0(t) \mid G(t)\big) &=
\left(\frac{2}{L}\Big(\sum_{i > 0}iY_{i,1}(t)+\sum_{j > 0}jY_{1,j}(t)-Y_0(t)\Big)-1\right)(1+o(1)).
\end{align*}
If we now interpolate the variables $L(t)$ and $Y_{ij}(t)$, by defining the scaled versions $y_{i,j}(t)=\frac{1}{n}Y_{i,j}(tn)$ for $i,j > 0$,
$\ell(t)=\frac{1}{n}L(tn)$ and $y_0(t)=\frac{1}{n}Y_0(tn)$, (to understand the limiting behavior as $n$ tends to infinity) this suggests the following system of
differential equations:
\begin{align*}
\frac{d\ell}{dt} & =  -3 \\
\frac{dy_{i,j}}{dt} &= \frac{2}{\ell}\big((i+1)y_{i+1,j}+(j+1)y_{i,j+1}-(i+j)y_{i,j}\big), \qquad i,j > 0, \\
\frac{dy_0}{dt} &= \frac{2}{\ell}\Big(\sum_{i > 0} iy_{i,1}+\sum_{j > 0} jy_{1,j}-y_0\Big)-1
\end{align*}
with the initial conditions $\ell(0)=3\gamma$, $y_0(0)=\sum_{i > 0} i \; d(i,0)+\sum_{j > 0} j \; d(0,j)$, and for $i,j > 0$ we have $y_{i,j}(0)=d(i,j)$, where
$d(i,j)$ is the scaled number of variables of $\conf$ appearing $i$ times positively and $j$ times negatively. It can be easily seen that $\ell(t)=3\gamma-3t$. Furthermore, defining the auxiliary function $b(t)=(1-t/\gamma)^{2/3}$, it can be checked,
that a solution (its uniqueness will be proven below) for the system of differential equations involving $y_{i,j}$ is the following function:
$$
y_{i,j}(t)=\sum_{k \geq i}\sum_{\ell \geq j} d(k,l) \binom{k}{i}\binom{\ell}{j}(b(t))^{i+j}(1-b(t))^{k-i}(1-b(t))^{\ell-j}.
$$
Also, it follows then that $$y_0(t)=3\gamma-3t-\sum_{i,j > 0} (i+j)y_{i,j}(t).$$
We are interested in the smallest value of $t$ for which $y_0(t/n)=0$. This value of $t$ will be the stopping time $T_\cD$ of our process (for our particular system of differential equations, starting with $\conf\in\Cnd$ with $\db\in\cN(n,\mbf\delta,\epsilon_1)$, we empirically observed that $T_\cD/n$ is between 0.15 and 0.16, and hence the resulting configuration $\conft$ a.a.s. still contains many literal occurrences. This observation is used to simplify the statement of Wormald's theorem).

To prove the uniqueness of the solution and the concentration of the resulting distribution, we use the following version of Wormald's theorem tailored to our specific system of ODE's:
\begin{thm}[Wormald~\cite{wormald99}]
Suppose that for any value of $G(t)$ and any $Y_{i,j}(t)$ (throughout the statement of the theorem $Y_{i,j}(t)$ also includes the case of $Y_0(t)$) we have
\[
\ex[Y_{i,j}(t+1)-Y_{i,j}(t) \mid G(t)] = f_{i,j}(t/n,\{Y_{i,j}(t)/n\}_{i,j})+O(1/n),
\]
where the $f_{i,j}: \mathbb{R}^{a} \rightarrow \mathbb{R}$ ($a$ being equal to the number of variables $Y_{i,j}$ plus $1$ for the variable $t$) are continuous functions with all Lipschitz constants uniformly bounded. Suppose also that $|Y_{i,j}(t+1)-Y_{i,j}(t)|\leq \beta$, and suppose that there is an order of the functions $f_{i,j}$ such that any $f_{i,j}$ depends on at most the terms preceding it in this order (including itself). Then, for some large enough constant $C$ and $0 < \alpha < 1/3$, the system of differential equations
\[
dy_{i,j}(x)/dx=f_{i,j}(x,\{y_{i,j}(x)\})
\]
with the initial conditions $y_{i,j}(0)=Y_{i,j}(0)/n$ has a unique solution and moreover, for any
$t \leq T_\cD - Cn^{1-\alpha}$, we have
\[
Y_{i,j}(t)=ny_{i,j}(t/n)+O(n^{1-\alpha})
\]
with probability at least $1-O(n^{2+\alpha}e^{-n^{1-3\alpha}})$.
\end{thm}
Defining the $f_{i,j}$'s to be the right hand sides of the system of differential equations of the $y_{i,j}$, i.e.,
\[
f_{i,j}(x,y_{i,j},y_{i,j+1},y_{i+1,j}) = \frac{2(-(i+j)y_{i,j}+(i+1)y_{i+1,j}+(j+1)y_{i,j+1})}{3\gamma-3x},
\]
and
\[
f_0(x,y_0,y_{1,1},y_{2,1},\ldots,y_{n,1},y_{1,2},\ldots,y_{1,n})=\frac{2}{3\gamma-3x}(-y_0+\sum_{i > 0} iy_{i,1}+\sum_{j > 0} jy_{1,j})-1,
\]
it is easily seen that these functions are continuous and satisfy a uniformly bounded Lipschitz constant.
Next, it is a routine calculation to show that for any $1 \leq i,j \leq n$
\[
|\ex(Y_{i,j}(t+1)-Y_{i,j}(t)\mid G(t)) - f_{i,j}\left(t/n, Y_{i,j}(t/n), Y_{i,j+1}(t/n),Y_{i+1,j}(t/n)\right)| = O(1/n)
\]
and
\[
|\ex(Y_{0}(t+1)-Y_{0}(t)\mid G(t)) - f_{0}\left(t/n, Y_0(t/n), \ldots,Y_{1,n}(t/n)\right)| = O(1/n).
\]
Also, we have
\[
\max_{1 \leq i,j \leq n} |Y_{i,j}(t+1)-Y_{i,j}(t)| \leq 2
\qquad\text{and}\qquad
|Y_0(t+1)-Y_0(t)| \leq 3.
\]
Moreover, we can rearrange the $f_{i,j}$'s in descending order of the sum of $i$ and $j$ (being $f_0$ the last element of this order), and hence we can apply Wormald's theorem. That is, for any $t \leq T_\cD - Cn^{1-\alpha}$, we have for any $i,j > 0$
\[
Y_{i,j}(t)=ny_{i,j}(t/n)+O(n^{1-\alpha})
\]
and
\[
Y_0(t)=ny_0(t/n)+O(n^{1-\alpha})
\]
with probability at least $1-O(n^{2+\alpha}e^{-n^{1-3\alpha}})$. Since in each step $Y_0$ as well as $Y_{i,j}$ change by at most $3$, the change of $Y_{i,j}$ and $Y_0$ caused by the eliminations of pure literal occurrences after time $T_\cD-Cn^{1-\alpha}$ is at most
$O(n^{1-\alpha})$, and thus the concentration results for $Y_{i,j}(t)$ and $Y_0(t)$ hold throughout the whole algorithm.

Now, since originally we start with $\conf\in\Cnd$ with $\db\in\cN(n,\mbf\delta,\epsilon_1)$, we have $|d(i,j)-\frac{e^{-3\gamma}(3\gamma/2)^{i+j}}{i!j!}| \leq \epsilon_1$ for any $i,j \in \nat$. Since the functions $f_{i,j}$ and $f_0$ appearing in the system of ODE's are continuous and their Lipschitz constants are uniformly bounded, by the continuous dependence theorem (see e.g.,\cite{ShwartzWeiss}, Theorem A.68) the solutions of the system of ODE's are continuous with respect to the initial conditions, uniformly for all $n$. That is, there exists an $\epsilon_2'> 0$, such that the degree sequence $\mbf \dt$ of $\conft$ satisfies $|\widetilde{d}(i,j)-\widetilde{\delta}(i,j)| \leq \epsilon_2'$  with probability at least $1-O(n^{2+\alpha}e^{-n^{1-3\alpha}})$, where (using the notation of $b$ and $\delta$ from above)
$$
\widetilde{\delta}(i,j)=\sum_{k \geq i}\sum_{\ell \geq j} \delta(k,\ell) \binom{k}{i}\binom{\ell}{j}(b(T_\cD/n))^{i+j}(1-b(T_\cD/n))^{k-i}(1-b(T_\cD/n))^{\ell-j}.
$$
When plugging in the initial condition $\delta(k,\ell)=e^{-3\gamma}\frac{(3\gamma/2)^{k+\ell}}{k!\ell!}$, the expression for $\widetilde{\delta}(i,j)$ can be simplified and we obtain
$$
\widetilde{\delta}(i,j)=e^{-3\gamma b}\frac{(3\gamma b/2)^{i+j}}{i!j!},
$$
where $b=b(t)=(1-t/\gamma)^{2/3}$.
Setting $\alpha=1/6$ and unbalancing $\conft$ (i.e., changing the syntactic sign of a variable appearing more often negatively than positively) we obtain a configuration $\confh$ whose degree sequence $\dbh$  satisfies for all $i,j \in \nat$
$$
|\widehat{d}(i,j)-\widehat{\delta}(i,j)| \leq \epsilon_2
$$
 with probability at least $1-O(e^{-\sqrt{n}})$, where
 $\epsilon_2=2\epsilon_2'$ and $\widehat{\delta}(i,j)$ is defined as follows:
 $$
\widehat{\delta}(i,j)=\left\{ \begin{array}{ll} 2\widetilde{\delta}(i,j), & \mbox{if } i > j, \\
\widetilde{\delta}(i,j), & \mbox{if } i=j, \\
0,& \mbox{if } i < j. \end{array}\right.$$ Setting $\beta=1/e$ concludes the proof of the lemma.
%
\section{Proof of Lemma~\ref{lem:clause}}
Let us fix $k\in\{0,1,2,3\}$. We will show that for a small enough $\epsilon_2>0$ the probability that $|\ch_k-\gammah_k|>\epsilon$ is $O(\beta^n)$ with $0<\beta<1$.
Given a random $\conf\in\Cnd$, we choose an ordering of the $\ch n$ clauses uniformly at random. For each $t=1,\ldots,\ch n$, define the random variable $Y_t$ to be $1$ if the $t$'\,th clause has syntactic type $k$ and $0$ otherwise. These random variables are all identically distributed, and we have that
\[
\pr(Y_t=1) = \frac{\binom{\ellp n}{k} \binom{\elln n}{3-k}}
{\binom{\ellp n+\elln n}{3}} = \binom{3}{k} \frac{[\ellp n]_{k} [\elln n]_{3-k}}
{[\ellp n+\elln n]_{3}}, \qquad \forall t, \; 1\le t\le\ch n.
\]
Let $S=\sum_{t=1}^{\ch n} Y_{t}$ and consider the following Doob martingale
\[
S_i=\ex \big(S \mid Y_{1},\ldots,Y_{i} \big), \qquad \forall i, \; 0\le i\le\ch n.
\]
We have $S_{\ch n}=S$ and also $S=\ch_k n$. Moreover,
\[
S_0=\ex S = \binom{3}{k}
\frac{[\ellp n]_{k} [\elln n]_{3-k}}{[\ellp n+\elln n]_{3}} \ch n
\sim \binom{3}{k} \frac{{\ellp}^{k} {\elln}^{3-k}}{(\ellp+\elln)^{3}} \ch n.
\]
Since $\dbh\in\cN(n,\deltabh,\epsilon_2)$, by choosing $\epsilon_2$ small enough, we can guarantee that $\ch$ is arbitrarily close to $\gammah$ and thus
\begin{equation}\label{eq:triangular}
|S_0/n-\gamma_k|<\epsilon/2.
\end{equation}
Moreover, by a case analysis depending on the syntactic type of the $(i+1)$'\,st clause (and depending on the value of $k$) it can be shown that $|S_{i+1}-S_{i}| \leq 2$ for any $i\in\{0,\ldots,\ch n\}$. Thus, by Azuma's inequality,
\[
\pr (|S_{\ch n} -S_{0}| \geq R) \leq e^{-\frac{R^2}{8\ch n}}.
\]
Setting $R=\epsilon n/2$, and by~\eqref{eq:triangular}, we get
\[
\pr (|\ch_k-\gamma_k| \geq \epsilon) =
\pr (|S_{\ch n}-\gamma_k n| \geq \epsilon n) \le
\pr (|S_{\ch n}-S_0| \geq \epsilon n/2) \le e^{-\Theta(n)}=\Theta(\beta^n),
\]
for some $0 < \beta < 1$.  Taking a union bound over all $k$ completes the proof of the lemma.
%
%
%
%
%
%
%
\section{Proof of Lemma~\ref{lem:boundary}}
We show the equivalent statement that $\log(F)$ does not maximize on the boundary of $\cP=\cP(n,\deltabh,\gammabh)$. 
Let $\bar x$ be a fixed point on the boundary of $\cP$. Since the boundary of $\cP$ consists only of non-negativity constraints of some parameters, this means that some values of $\ell_{\sigma}$, $t_{i,j}$, $f_{i,j,k}$, $c_{\alpha}$, $h_\tnsf$ or $h_\tnsr$ of $\bar x$ have to be zero. It can be shown that the only $\ell_\sigma$ which can be zero is $\ellnsr$: for all other $\ell_\sigma$ one can consider the linear program which contains all constraints and as objective function the minimization of the particular $\ell_\sigma$. It can be shown that the value of the objective function is in all cases except for $\ellnsr$ strictly greater than zero (using CPLEX). In the following we call all $c_\alpha$ for which $\tnsr(\alpha)>0$, all $f_{i,j,k}$ with $M \geq j > k \geq 1$ as well as $h_\tnsr$ (cf.~\eqref{eq:tfh-ell's2} and~\eqref{eq:c-ell's2}) to be \emph{forced}. 
We now make a case analysis depending on which of the remaining parameters is zero. \\
\textbf{Case 1: None of the $\ell_{\sigma}$ is equal to zero.}\\
In this case we know that at least one of the $t_{i,j}$, $f_{i,j,k}$, $c_{\alpha}$, $h_\tnsf$ or $h_\tnsr$ is equal to zero. First observe that the interior of $\cP$ contains at least one feasible point $\bar{x_0}$: since all constraints are linear, one can add additional linear constraints to ensure that all values of $\ell_{\sigma}$, $t_{i,j}$, $f_{i,j,k}$, $c_{\alpha}$, $h_\tnsf$ and $h_\tnsr$ are strictly positive and then check the feasibility of the constraints (we did this using CPLEX). Since $\cP$ is convex, we can find a direction that points towards $\bar{x_0}$. In this direction all the $t_{i,j}$, $f_{i,j,k}$, $c_{\alpha}$, $h_\tnsf$ and $h_\tnsr$ components which were zero must increase (perhaps at different rates), while the non-zero variables can either increase or decrease (and remain non-zero). Hence, the directional derivative in this direction contains at least one term which is $-\log(0)$ (i.e. $+\infty$)  (plus/minus some constants) and hence it is an increasing direction.\\
\textbf{Case 2: $\ellnsr=0$, and in addition to the forced values of $c_\alpha$, $f_{i,j,k}$ and $h_\tnsr$ at least one of the other values of $c_\alpha$, $t_{i,j}$, $f_{i,j,k}$ or $h_\tnsf$ is zero as well.}\\
Using a similar argument as in the previous case, we can find a feasible point $\bar{x_1}$ for which $\ellnsr=0$, all forced parameters are zero, but no other parameters are zero (we did this using CPLEX). Again, by convexity of $\cP$, we can find a direction that points towards $\bar{x_1}$. In this direction $\ellnsr$ and all the forced parameters remain zero, but all other $c_\alpha$, $t_{i,j}$, $f_{i,j,k}$ and $h_\tnsf$ which were zero before increase (there is at least one of them). Hence, the directional derivative in this direction contains at least one term which is $-\log(0)$ (plus/minus some constants), and hence it is an increasing direction.\\
\textbf{Case 3: $\ellnsr=0$, all the forced $c_\alpha$, $f_{i,j,k}$ and $h_\tnsr$ are zero, but nothing else is zero.}\\
In this case we move infinitesimally in the following direction: we increase $\ellnsr$ and decrease $\ellnsf$ by the same infinitesimally small amount $=:\zeta$, we increase $c_\alpha$ for $\alpha \in \left\{\mtrx{1}{0}{1}{1}, \mtrx{1}{0}{0}{2}, \mtrx{2}{0}{1}{0}\right\}$ by $\zeta$, we decrease $c_\alpha$ for $\alpha \in \left\{\mtrx{0}{1}{1}{2}, \mtrx{1}{0}{1}{1}, \mtrx{3}{0}{0}{0}\right\}$ by $\zeta$, and we also increase $f_{2,2,1}$ and $h_\tnsr$ by $\zeta/2$ and decrease $f_{2,2,2}$ and $h_\tnsf$ by $\zeta/2$ (note that since all values of $f_{i,j,k}$, $c_\alpha$ and $h_\tnsf$ which are not forced are strictly positive, all the changes are allowed). It can be checked that when performing all these changes, all constraints remain valid, and by convexity of $\cP$, we can move into this direction. In this direction the variables $\ellnsr$, $c_\alpha$ with $\alpha = \mtrx{1}{0}{1}{1}$, $f_{2,2,1}$ and $h_\tnsr$ which were zero before move away from zero, all other zero variables remain zero, and all other non-zero variables remain non-zero. Hence, the directional derivative in this direction contains  one term $\log(0)$ and three terms $-\log(0)$ (plus/minus some constants), and hence it is an increasing direction.

%
%
%
%
%
\section{Critical points of $\mbf{\log F}$}\label{sec:lagrange}
In order to characterize the critical points of $\log F$ and the value of the maximum of $F$, we need some definitions. Let us consider the functions
\begin{align*}
&\Psi_{i,j} = {\nu_{\tps}}^i + (\nu_{\tnsf}+\nu_{\tnsr})^j - {\nu_{\tnsr}}^j \qquad \forall(i,j)\in\cL,
\\
&\Psi_k = \sum_{\substack{\alpha\in\cA\\ \tp(\alpha)=k}} \frac{2}{w(\alpha)!} {\mu_{\tps}}^{\tps(\alpha)} {\mu_{\tnsf}}^{\tnsf(\alpha)} {\mu_{\tnsr}}^{\tnsr(\alpha)} \qquad \forall k\in\{0,\ldots,3\},
\\
&\Psi = \nu_\tnsf + \nu_\tnsr, \qquad
\Psi_\tp = \nu_\tps\mu_\tps + 1, \qquad
\Psi_\tn = \nu_\tnsf\mu_\tnsf + \nu_\tnsr\mu_\tnsr + 1,
\end{align*}
defined on tuples $(\bar\nu,\bar\mu)=(\nu_\tps,\nu_\tnsf,\nu_\tnsr,\mu_\tps,\mu_\tnsf,\mu_\tnsr)$ of positive reals, and let
\[
L = {\nu_\tps}^{h_\tps} \left(\prod_{\cL}{\Psi_{i,j}}^{\deltah_{i,j}}\right) \Psi^{h_\tns} \left(\prod_{k=0}^3 {\Psi_k}^{\gammah_k}\right) {\Psi_\tp}^{-\lap} {\Psi_\tn}^{-\lan}
\]
\begin{lem}\label{lem:lagrange}
There is a bijective correspondence between critical points of $\log F$ in the interior of $\cP(n,\deltabh,\gammabh)$ and positive solutions of the $6\times6$ system of equations
\begin{equation}
\nabla_{\bar\nu,\bar\mu} (\log L) = 0.
\label{eq:lagrange_gradient}
\end{equation}
Moreover, the value of $F$ at each of these critical points can be expressed in terms of the corresponding solution by
\begin{equation}
\nu_{\tps}^{h_\tps} \left(\frac{\lap}{\Psi_\tp}\right)^{\lap} \left(\frac{\lan}{\Psi_\tn}\right)^{\lan} \prod_{\cL}\left(\frac{\Psi_{i,j}}{\deltah_{i,j}}\right)^{\deltah_{i,j}} \left(\frac{\Psi}{h_\tns}\right)^{h_\tns} \prod_{k=0}^3 \left(\frac{\Psi_k}{\gammah_k}\right)^{\gammah_k}
\label{eq:optimum_lagrange}
\end{equation}
\end{lem}
\begin{proof}
We use the Lagrange multipliers technique to characterize the critical points of $\log F$ in the interior of $\cP(n,\deltabh,\gammabh)$. Let $(\rho'_{i,j})_\cL$, $\rho'$, $(\rho'_k)_{k\in\{0,\ldots,3\}}$ and $(\rho'_\tp,\rho'_\tn)$ be the Lagrange multipliers of the equations in~\eqref{eq:tf's}, \eqref{eq:h's}, \eqref{eq:c's} and~\eqref{eq:ell's} respectively. These equations are called \emph{partition constraints}.
Moreover, let $(\nu'_{\tps},\nu'_{\tnsf},\nu'_{\tnsr})$ and $(\mu'_{\tps},\mu'_{\tnsf},\mu'_{\tnsr})$ be the Lagrange multipliers of the equations in~\eqref{eq:tfh-ell's2} and~\eqref{eq:c-ell's2} respectively. After a few manipulations, the standard Lagrange multiplier equations can be written as
\begin{align}
&t_{i,j} = \rho_{i,j} {\nu_{\tps}}^i \quad \forall(i,j)\in\cL, \qquad
f_{i,j,k} = \binom{j}{k} \rho_{i,j} {\nu_{\tnsf}}^k {\nu_{\tnsr}}^{j-k} \quad \forall(i,j,k)\in\cL',
\notag\\
&h_{\tnsf} = \rho \nu_{\tnsf}, \qquad
h_{\tnsr} = \rho \nu_{\tnsr},
\notag\\
&c_\alpha = \frac{2}{w(\alpha)!} \rho_{\tp(\alpha)} {\mu_{\tps}}^{\tps(\alpha)} {\mu_{\tnsf}}^{\tnsf(\alpha)} {\mu_{\tnsr}}^{\tnsr(\alpha)} \quad \forall\alpha\in\cA,
\notag\\
&\ell_{\tps} = \rho_{\tp} \nu_{\tps} \mu_{\tps}, \qquad
\ell_{\tpu} = \rho_{\tp}, \qquad
\ell_{\tnsf} = \rho_{\tn} \nu_{\tnsf} \mu_{\tnsf}
\notag\\
&\ell_{\tnsr} = \rho_{\tn} \nu_{\tnsr} \mu_{\tnsr}, \qquad
\ell_{\tnu} = \rho_{\tn},
\label{eq:lagrange1}
\end{align}
where
\begin{align*}
&\rho_{i,j} = e^{\rho'_{i,j}-1} \quad \forall(i,j)\in\cL, \qquad
\rho = e^{\rho-1},
\\
&\rho_k = e^{\rho'_k-1} \quad \forall k\in\{0,\ldots,3\}, \qquad
\rho_\sigma = e^{-\rho'_\sigma-1} \quad \forall\sigma\in\{\tp,\tn\},
\\
&\nu_\sigma = e^{\nu'_\sigma} \quad \forall\sigma\in\{\tps,\tnsf,\tnsr\},
\qquad
\mu_\sigma = e^{\mu'_\sigma} \quad \forall\sigma\in\{\tps,\tnsf,\tnsr\}.
\end{align*}
Notice that the new variables defined above must be strictly positive. By combining~\eqref{eq:lagrange1} with the partition constraints~\eqref{eq:tf's}, \eqref{eq:h's}, \eqref{eq:c's} and~\eqref{eq:ell's}, the variables $(\rho_{i,j})_\cL$, $\rho$, $(\rho_k)_{k\in\{0,\ldots,3\}}$ and $(\rho_\tp,\rho_\tn)$ are eliminated and $(\bar t,\bar f,\bar h,\bar c,\bar\ell)$ can be expressed just in terms of $\nu_\sigma$ and $\mu_\sigma$ ($\sigma\in\{\tps,\tnsf,\tnsr\}$):
\begin{align}
&t_{i,j} = \deltah_{i,j}\frac{{\nu_{\tps}}^i}{\Psi_{i,j}} \quad \forall(i,j)\in\cL, \qquad
f_{i,j,k} = \deltah_{i,j}\frac{\binom{j}{k} {\nu_{\tnsf}}^k {\nu_{\tnsr}}^{j-k}}{\Psi_{i,j}} \quad \forall(i,j,k)\in\cL'
\notag\\
&h_{\tnsf} = h_\tns \frac{\nu_{\tnsf}}{\Psi}, \qquad
h_{\tnsr} = h_\tns \frac{\nu_{\tnsr}}{\Psi},
\notag\\
&c_\alpha = \gammah_{\tp(\alpha)} \frac{\frac{2}{w(\alpha)!} {\mu_{\tps}}^{\tps(\alpha)} {\mu_{\tnsf}}^{\tnsf(\alpha)} {\mu_{\tnsr}}^{\tnsr(\alpha)}}{\Psi_{\tp(\alpha)}} \quad \forall\alpha\in\cA,
\notag\\
&\ell_{\tps} = \lap \frac{\nu_{\tps} \mu_{\tps}}{\Psi_\tp}, \qquad
\ell_{\tpu} = \lap \frac{1}{\Psi_\tp}, \qquad
\ell_{\tnsf} = \lan \frac{\nu_{\tnsf} \mu_{\tnsf}}{\Psi_\tn},
\notag\\
&\ell_{\tnsr} = \lan \frac{\nu_{\tnsr} \mu_{\tnsr}}{\Psi_\tn}, \qquad
\ell_{\tnu} = \lan \frac{1}{\Psi_\tn}
\end{align}
Then we plug these expressions into the constraints~\eqref{eq:tfh-ell's2} and~\eqref{eq:c-ell's2} and after a few manipulations we obtain the following $6\times6$ system on variables $\nu_\sigma$, $\mu_\sigma$
\begin{subequations}
\begin{gather}
\begin{aligned}
\nu_\tps \left(\lap\frac{\partial\Psi_\tp/\partial_{\nu_\tps}}{\Psi_\tp}\right) &= \nu_\tps \left( \sum_{\cL} \deltah_{i,j}\frac{\partial\Psi_{i,j}/\partial_{\nu_\tps}}{\Psi_{i,j}}\right) + h_\tps,
\\
\nu_\tnsf \left( \lan\frac{\partial\Psi_\tn/\partial_{\nu_\tnsf}}{\Psi_\tn}\right) &= \nu_\tnsf \left( \sum_{\cL} \deltah_{i,j}\frac{\partial\Psi_{i,j}/\partial_{\nu_\tnsf}}{\Psi_{i,j}} + h_\tns\frac{\partial\Psi/\partial_{\nu_\tnsf}}{\Psi}\right),
\\
\nu_\tnsr \left( \lan\frac{\partial\Psi_\tn/\partial_{\nu_\tnsr}}{\Psi_\tn}\right) &= \nu_\tnsr \left( \sum_{\cL} \deltah_{i,j}\frac{\partial\Psi_{i,j}/\partial_{\nu_\tnsr}}{\Psi_{i,j}} + h_\tns\frac{\partial\Psi/\partial_{\nu_\tnsr}}{\Psi}\right),
\end{aligned}
\label{eq:6x6nu}
\\
\begin{aligned}
\mu_\tps \left(\lap\frac{\partial\Psi_\tp/\partial_{\mu_\tps}}{\Psi_\tp}\right) &= \mu_\tps \left( \sum_{k=0}^3 \gammah_k\frac{\partial\Psi_k/\partial_{\mu_\tps}}{\Psi_k}\right),
\\
\mu_\tnsf \left(\lan\frac{\partial\Psi_\tn/\partial_{\mu_\tnsf}}{\Psi_\tn}\right) &= \mu_\tnsf \left( \sum_{k=0}^3 \gammah_k\frac{\partial\Psi_k/\partial_{\mu_\tnsf}}{\Psi_k}\right),
\\
\mu_\tnsr \left(\lan\frac{\partial\Psi_\tn/\partial_{\mu_\tnsr}}{\Psi_\tn}\right) &= \mu_\tnsr \left( \sum_{k=0}^3 \gammah_k\frac{\partial\Psi_k/\partial_{\mu_\tnsr}}{\Psi_k}\right),
\end{aligned}
\label{eq:6x6mu}
\end{gather}
\end{subequations}
which can rewritten as in~\eqref{eq:lagrange_gradient}.
\end{proof}
By restricting our attention to~\eqref{eq:6x6mu}, we observe that each of the variables $\nu_\tps$, $\nu_\tnsf$ and $\nu_\tnsr$ appears only in one of the three equations, and also that we can easily express $\nu_\tps$, $\nu_\tnsf$ and $\nu_\tnsr$ in terms of $\mu_\tps$, $\mu_\tnsf$ and $\mu_\tnsr$. By plugging this expression into~\eqref{eq:6x6nu}, we obtain a $3\times3$ system in terms of the variables $\mu_\tps$, $\mu_\tnsf$ and $\mu_\tnsr$ (recall that only positive solutions of the $\mu_\sigma$ which lead to positive $\nu_\sigma$ are considered). We numerically solved this $3\times3$ system with the help of Maple, using $50$ digits of precision and $M=23$. The unique solution obtained by Maple is
\begin{align*}
\mu_\tps &\approx 1.9972796, & \mu_\tnsf &\approx 0.45029358, &
\mu_\tnsr &\approx 0.33794030,
\\
\nu_\tps &\approx 1.2782018, & \nu_\tnsf &\approx 0.33277280, &
\nu_\tnsr &\approx 0.95336927.
\end{align*}
In view of~\eqref{eq:bound_expectation} and Lemma~\ref{lem:lagrange}, we evaluate~\eqref{eq:optimum_lagrange} at these values and multiply this by $B$ given in~\eqref{eq:B}, and we obtain the bound
\begin{equation}
\ex X \le \left( (1+ 10^{-7}) 0.9999998965 \right)^n.
\end{equation}

\end{document}